\documentclass[a4paper,USenglish,nolineno]{socg-lipics-v2019}
\usepackage{amsthm,url,xspace,cite,pdfpages,thmtools,thm-restate}
\nolinenumbers

\newcommand\e\emph\renewcommand\d{\ensuremath{\Delta_r}\xspace}\renewcommand\L{\ensuremath{L}\xspace}\newcommand\p{\ensuremath{p}\xspace}\renewcommand\P{\ensuremath{P}\xspace}\newcommand\s{\ensuremath{s}\xspace}\renewcommand\t{\ensuremath{t}\xspace}\newcommand\V[1]{\ensuremath{V}({#1})}\renewcommand\sp{{\sc Secluded Path}\xspace}\newcommand\eps{\ensuremath\varepsilon\xspace}
\title{Geometric secluded paths and planar satisfiability}
\author{Kevin Buchin}{Department of Mathematics and Computer Science, TU Eindhoven, the Netherlands}{k.a.buchin@tue.nl}{}{}
\author{Valentin Polishchuk and Leonid Sedov}{Communications and Transport Systems, ITN, Link\"oping University, Sweden}{firstname.lastname@liu.se}{}{}
\author{Roman Voronov}{Institute of Mathematics and Information Technologies, Petrozavodsk State University, Russia}{rvoronov@petrsu.ru}{}{}
\authorrunning{K. Buchin, V. Polishchuk, L. Sedov, R. Voronov}
\Copyright{Kevin Buchin, Valentin Polishchuk, Leonid Sedov, Roman Voronov}

\begin{document}\maketitle
\begin{abstract}We consider paths with low \emph{exposure} to a 2D polygonal domain, i.e., paths which are seen as little as possible; we differentiate between \emph{integral} exposure (when we care about how long the path sees every point of the domain) and \emph{0/1} exposure (just counting whether a point is seen by the path or not). For the integral exposure, we give a PTAS for finding the minimum-exposure path between two given points in the domain; for the 0/1 version, we prove that in a simple polygon the shortest path has the minimum exposure, while in domains with holes the problem becomes NP-hard. We also highlight connections of the problem to minimum satisfiability and settle hardness of variants of planar min- and max-SAT.
\end{abstract}
\begin{table}\centering\begin{minipage}[b]{.4\columnwidth}\centering\begin{tabular}{|c|c|c|}\hline
& $0/1$ & Integral \\
&exposure&exposure\\\hline
With holes & Hard (Thm.~\ref{tmin})  & PTAS \\\cline{1-2}
Simple &P (Thm.~\ref{tsimple}) & (Thm.~\ref{t:ptas})\\\hline\end{tabular}\caption{Hardness of minimum-exposure paths in polygonal environments}\label{results:secluded}\end{minipage}\hfill
\begin{minipage}[b]{.5\columnwidth}\centering\begin{tabular}{|c|c|c|}\hline
& min2SAT & max2SAT \\\hline
V-cycle & Hard (Thm.~\ref{t:opt2sat}) & \\\cline{1-2}
VC-cycle & Hard (Thm.~\ref{t:opt2sat}) & Hard\\\cline{1-2}
Separable & P (Thm.~\ref{t:separable:min}) & (Thm.~\ref{t:opt2sat})\\\cline{1-2}
Monotone & P (Cor.~\ref{c:monotone:min}) & \\\hline
\end{tabular}\caption{Hardness of versions of planar opt2SAT (see Section~\ref{2sat} for definitions)}\label{results:opt2sat}\end{minipage}\end{table}
\section{Introduction and Related work}
Both visibility and motion planning are textbook subjects in computational geometry – see, e.g., the respective chapters in the handbook \cite{handbook} and the books \cite{agtBook,ghoshBook}. Visibility meets routing in a variety of geometric computing tasks. Historically, the first approach to finding shortest paths was based on searching the visibility graph of the domain; visibility is vital also in computing \e{minimum-link} paths, i.e., paths with fewest edges \cite{mrw,suri,minlink,minlink3d}. ''Visibility-driven'' path planning has attracted also some recent interest \cite{qvm,Haitao,gender}. In addition to the theoretical considerations, visibility and motion planning are closely coupled in practice: computer vision and robot navigation go hand-in-hand in many courses and real-world applications.

The question of \emph{hiding} a path in a polygonal domain was first raised in a SoCG'88 paper \cite{01socg}: it considered the \emph{robber route problem} in which the goal is to minimize the length traveled within sight of at least one of a number of \emph{threats} (each threat being a point); the problem reduces to finding the shortest path in the $0/1/\infty$ metric that assigns a cost of 1 to the union of the visibility polygons of the threats, and 0 to the rest of the domain (and infinite weight to the complement of the domain, where travel is forbidden). Our settings are different from \cite{01socg} in two aspects: (1)~we have a \e{continuum} of the threats (every point in the domain is a threat) and (2)~in the integral version, we care for how long threats are seen from points along the path (formally: we integrate the visible area along the path); in other words, we account for the ``intensity'' of the visibility from the threats.

Lately, motivated by the rise of the Internet of things (IoT) and mobile computing, there has been a surge of research on anonymity, security, confidentiality and other forms of privacy preservation (in particular, in geometric environments~\cite{mobihoc}), studying paths with minimum exposure to sensors in a network~\cite{sensor,sensor1,sensor2,sensor3}. The standard model, again, assumes a finite number of point sensors, so the visibility is changing discretely, as the path goes in/out of a sensor coverage. To our knowledge, Lebeck, M{\o}lhave and Agarwal~\cite{occluded1,occluded2} were the first to introduce integration of the visibility \e{continuously} changing along the path (which is also one of our models). 
Our paper is different from Lebeck et al.\ in that we give algorithms with provable theoretical performance in continuous domains under the usual notion of distance-independent visibility. Lebeck et al.\ presented strategies with outstanding practical performance on discretized terrains, in the more realistic model of visibility deteriorating with distance.

Minimizing the integral exposure can be viewed as an extension of the \e{weighted region problem (WRP)}~\cite{algebraic,wrp,inkulu,papadimitriou,aleksandrov,anisotropic,querying,refinement} to the case of continuously changing weight, where the weight of a point is the area of its visibility polygon; in the WRP the input is a weighted polygonal subdivision of the domain (with a constant weight assigned to each cell of the subdivision) and the goal is to find the path minimizing the integral of the weight along the path. The computational complexity of the WRP is open; PTASs for the problem have running times that depend not only on the complexity of the subdivision, but also on various parameters of the input like ratio of max/min weight, largest coordinate and angles of the regions, etc.\ (the parameters differ between the algorithms, see \cite[Ch.~31]{handbook3} for details). Integration of other measures of ``local quality'' (different from visibility) for points along a path was the subject also in the study of \e{high-quality} paths~\cite{quality1,quality2} and related research~\cite{vv,jurWAFR06}.

Recent papers~\cite{chechik,kulikov,luckow,novosibirsk} explored paths adjacent to few vertices in graphs; such paths were dubbed \e{secluded} in \cite{chechik}. Our paper may thus be viewed as studying geometric versions of the secluded path problem.

\subparagraph*{Contributions and Roadmap} In Section~\ref{smin} we prove that in a polygonal domain with holes it is NP-hard to find a path, between two given points, minimizing the area seen from the path; the reduction is from minSAT (find the truth assignment to Boolean variables so as to satisfy the minimum number of given disjunctive clauses). In Section~\ref{ssimple} we complement the hardness by showing that in a simple polygon, shortest paths are the ones that see minimum area; even more generally, we prove that in a polygon with holes, a locally shortest path sees less area than any path of the same homotopy type (because for a small number of holes the homotopy types can be efficiently enumerated, this implies that the problem is FPT parameterized by the number of holes). Section~\ref{ptas} gives a PTAS for minimizing the integral of the seen area along the path; we first give a generic scheme for building a piecewise-constant approximation of the visibility area for points in the domain, and then in Section~\ref{implementation} present details of an implementation which allows applying a PTAS for WRP on our ``pixels'' with approximately constant seen area. Finally, in Section~\ref{2sat} we further explore the connection between path hiding and minSAT, and determine hardness of versions of planar minSAT (and maxSAT).

Tables~\ref{results:secluded} and~\ref{results:opt2sat} summarize the results. We leave open designing an approximation algorithm for minimizing the seen area, as well as the complexity of the integral version of the problem.

\subparagraph*{Notation and Problems formulation} We use $|\cdot|$ to denote the measure of a set, i.e., length of a segment and area of a 2D set. Let \P be a polygonal domain with $n$ vertices and $\s,\t\in\P$ be two given points in it. For a point $p\in\P$ let $\V\p\subseteq\P$ be the visibility polygon of \p, i.e., the set of points seen by \p. We study the following problems:
\begin{itemize}
\item {\sc Geometric Secluded Path}: Find the \s-\t path that sees as little area of \P as possible (the area seen by a path is defined as the area seen by at least one point of the path, i.e., the so called \e{weak} visibility region of the path).
\item {\sc Integral Geometric Secluded Path}: Find the \s-\t path $\pi$ that minimizes the integral of the area of the visibility polygon over the points along the path, $\int_{\pi}|\V\p|\,\mathrm{d}p$.
\end{itemize}

\section{Minimizing seen area}\label{smin}
We prove that exposure minimization is NP-hard in general, but in simple polygons the minimum-exposure path is the shortest path.
\begin{theorem}\label{tmin}{\sc Geometric Secluded Path} is NP-hard.
\end{theorem}
\begin{proof}
We reduce from min2SAT: find truth assignment for a set of $n$ variables, satisfying the minimum number of given two-literal disjunctive clauses. (Inside this proof $n$ will denote the number of variables and $c$ the number of clauses.) Figure~\ref{reduction}, left illustrates the construction. A variable gadget is an isosceles triangle. The triangles for the variables are stacked into a \emph{Christmas tree}, with \s and \t placed at the top and the root respectively. Going through the left (resp.\ right) vertex of a triangle represents setting the variable to True (resp.\ False). The clauses are all put on a horizontal line above the Christmas tree so that the segment between any literal and any clause does not intersect the tree. Each clause is connected to its literals, and all connections (including the ones forming the Christmas tree edges) are thin corridors forming the domain; a clause gadget is simply the intersection of the two corridors. The idea of the reduction is to have an \s-\t path go through all variable gadgets, choosing whether to go through the variable or its negation in every gadget: the fewer clause gadgets are seen, the fewer clauses are satisfied.
\begin{figure}\centering\includegraphics{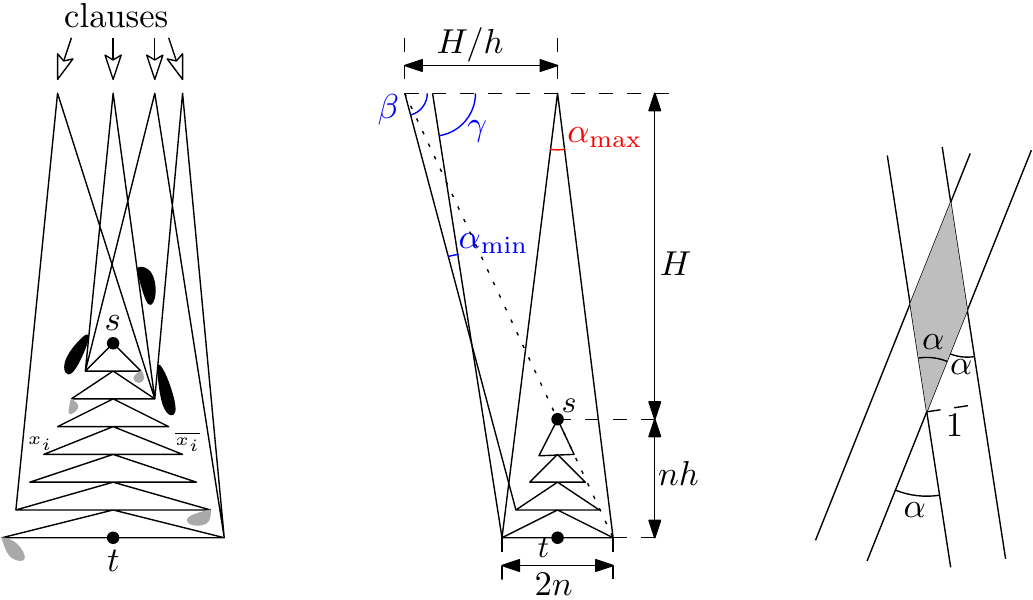}
\caption{Left: The reduction from min2SAT. All segments are thin corridors of~\P. Some leakage-blocking high-area chambers are shown black and some area equalizers are shown gray (both black and gray belong to the domain). Middle: the largest angle at a clause and the smallest angle at a midway intersection. The $c$ clauses are spread evenly on the segment of width $2H/h$; thus, the distance between clauses (the base of the triangle with angle $\alpha_{\min}$ at the apex) is $\frac{2H}{h(c-1)}$. Right: Midway intersection of unit-width corridors is area-$\frac1{\sin\alpha}$ rhombus with side $\frac1{\sin\alpha}$ and angle $\alpha$.}\label{reduction}
\end{figure}

A few technicalities have to be taken care of:\begin{itemize}\item Two variable--clause corridors, leading to different clauses, may intersect \emph{midway}, meaning that the intersection area may be seen twice. We have to make sure that the area of such a midway intersection is much smaller than the clause gadget area. Being smaller by a factor $4c^3$ will suffice: even if parts of a corridor are seen due to the midway intersections with all (at most $2c-1$) other corridors, the total seen corridor's midway area will still be smaller (by a factor $\approx2c$) than the area of a single clause gadget. Moreover, with such small midway intersections, they may be neglected altogether when counting the areas of clause gadgets seen from literals: the total areas of all (at most $4c^2$ midway intersections) will be smaller by at least a factor of $c$ than the area of a single clause gadget. To reduce areas of the midway intersections in comparison to the clause gadgets areas, we put the clause gadgets high above the Christmas tree -- at height $H$, to be determined later (Fig.~\ref{reduction}, middle). The area of intersection of two corridors (Fig.~\ref{reduction}, right) is inversely proportional to the sine of the angle between the corridors (the corridors are all of the same width), so the smallest-area clause gadget would be the one for the clause $x_n\lor\overline{x_n}$ placed directly above the apex of the Christmas tree (since we do not control which clause goes where on the clauses line, we have to consider the worst case); let $\alpha_{\max}$ be the angle between the corridors defining the gadget. Assuming the height of every variable gadget triangle is $h$ and their bases have lengths $2,4,\dots,2n$ (refer to Fig.~\ref{reduction}, middle),
\[
\alpha_{\max}=2\arctan\frac{n}{H+nh}.
\]
On the other hand, the \e{smallest angle} between two interesting corridors that do not lead to the same clause (i.e., the smallest angle that may define the area of a midway intersection) can be formed by corridors leading to last and last-but-one clause from the last-but-one and last variables $x_n,x_{n-1}$ resp.\ (changing the endpoints of the corridors would only increase the angle of intersection); the angle is
\[
\alpha_{\min}=\gamma-\beta=\arctan\frac{H+nh}{H/h-\frac{2H}{h(c-1)}-n}-\arctan\frac{H+(n-1)h}{H/h-(n-1)}.
\]
By trigonometric formulas, the ratio $\sin\alpha_{\min}/\sin\alpha_{\max}$, after being squared a constant number of times, is a ratio of polynomials. A Mathematica script shows that this ratio tends to infinity as $H$ grows (Appendix~\ref{listing} gives the Mathematica listing); hence, at a polynomially large $H$, the ratio becomes larger than $4c^2$, as we need.
\item We make sure that the area, around a clause gadget, seen from one literal but not from the other (Fig.~\ref{technicalities}, left), is negligible in comparison with the clause gadget area (seen from both literals of the clause). This is already taken care of by the above, as the whole construction is made tall (large~$H$).
\item Leakage of paths from the Christmas tree into variable--clause corridors is prevented by attaching a large-area chamber to each corridor (between the literal and the first intersection of two corridors), so that a path going through the corridor would see the whole area of the chamber. To ensure that the area of a single chamber is larger than the area seen by any path through the Christmas tree, the whole construction is scaled up while keeping the width of the corridors fixed: since the areas available for the chambers grow quadratically with the scaling factor and the areas seen along the corridors grow linearly, a polynomial scaling will suffice to ensure that the chambers areas are large enough to prevent the path going anywhere except through the variable gadgets.
\item We attach area equalizers to the literals so that no matter whether the path passes through the variable or its negation, it sees the same non-clause area (the areas may be different between the different variables; we only make sure that for any single variable the seen non-clause area does not depend on whether the variable is set to true or false by the path).
\item In the construction so far, different clause gadgets may have different areas; let $a$ denote the smallest area of a clause gadget. We make sure that all clause gadgets have area $a$, which can be done e.g., by appropriately cutting off the clause gadgets from the top (Fig~\ref{technicalities}, right).
\end{itemize}
Now, all \s-\t paths, going through the Christmas tree only, will see the same non-clause area~$A$. The total area seen by a path is then $\approx A+ka$ where $k$ is the number of clauses seen by the path, which is the same as the number of clauses satisfied by the truth assignment set by the path (we say that the seen area is \e{approximately} equal to $A+ka$ because of the non-counted areas that may be seen---midway intersections and parts seen by one literal only---which we made sure to be negligible in comparison with~$a$).
\end{proof}

\begin{figure}\centering
\begin{minipage}[b]{.6\columnwidth}\centering\includegraphics{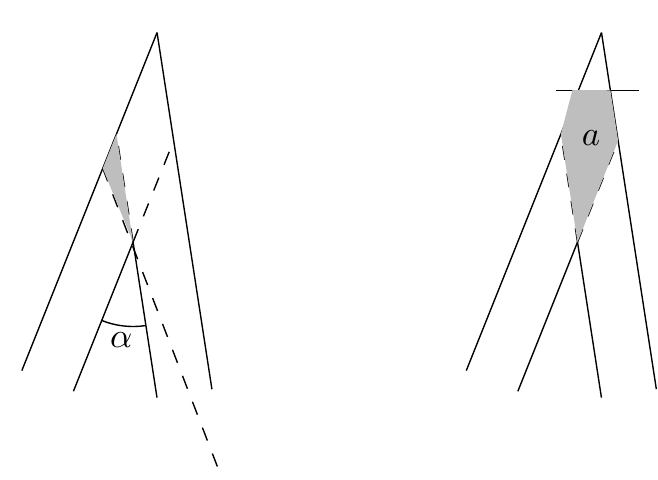}\caption{Left: Area seen by one literal only (gray) is negligible for small $\alpha$. Right: Decreasing clause gadget area.}\label{technicalities}\end{minipage}
\hfil
\begin{minipage}[b]{.3\columnwidth}\centering\includegraphics{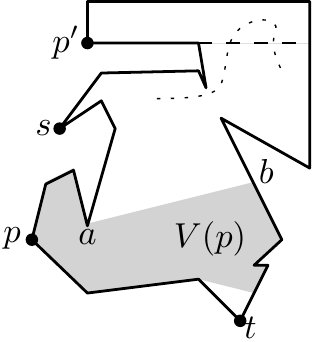}\caption{\V\p\ is shaded; $ab$ is the essential cut of \p. A dotted path, crossing the cut of $p'$ (dashed), can be shortcut along the cut.}\label{simple}\end{minipage}
\end{figure}

In Section~\ref{2sat} we discuss why we could not use \e{planar} min2SAT to prove hardness of {\sc Geometric Secluded Path}, avoiding dealing with the crossings.

\subsection{Simple polygons}\label{ssimple}
We show that in a simple polygon shortest paths see least area:
\begin{theorem}\label{tsimple}
If \P is a simple polygon, the shortest \s-\t path is the solution to \sc{Geometric Secluded Path}.
\end{theorem}
\begin{proof}The visibility polygon \V\p\ of a point $p\in P$ is bounded by edges and chords of \P, with each chord connecting a vertex of the polygon to a point on its boundary. If \P is a simple polygon and \p does not see \s ($\s\notin\V\p$), then there is a unique chord separating \p from \s; the chord is called the \emph{essential cut} of \p~\cite{bengt} (Fig.~\ref{simple}).

If an essential cut does not separate \s from \t, then the shortest \s-\t path does not cross the cut, for otherwise, the path could be shortcut along the cut. That is, the shortest path crosses those and only those cuts that separate \s from \t. But any other path also has to cross all such cuts, i.e., has to see all the points seen by the shortest path.
\end{proof}

For polygons with a small number of holes one may go through all homotopy types of simple (without self-intersections) \s-\t paths: a simple argument (Lemma~\ref{homotopy} in the appendix) shows that a shortcut of a path sees less than the original path, and hence the locally shortest path is the secluded path within its homotopy class. 

\section{A PTAS for minimizing integral exposure}\label{ptas} 
In Section~\ref{generic} we give a generic way to partition the domain in such a way that the visible area is approximately constant within a cell of the partition; then in Section~\ref{implementation} we present details of a slightly different partitioning, having straight-line edges, on which a PTAS for the WRP can be applied to find the path with approximately minimum integral exposure.

\subsection{Reduction to WRP with curved regions}\label{generic} 
We first compute the \e{visibility graph} of \P, i.e., the graph connecting pairs of mutually visible vertices of the domain, and extend every edge of the graph in both directions maximally within \P. The extensions of the visibility edges split \P into $O(n^4)$ cells such that the visibility polygon \V\p\ is combinatorially the same for any point \p within one cell of the subdivision; the subdivision is called the \e{visibility decomposition} of \P \cite{decomposition}. In particular, the area |\V\p| is given by the same formula for any point \p in one cell $\sigma$ of the decomposition.
Specifically, the rays from \p through the seen vertices of \P split \V\p\ into $O(n)$ triangles (Fig.~\ref{V}, left). The side of any triangle, opposite to \p, is a subset of an edge of \P; we call this side the \e{base} of the triangle. Each of the other, non-base sides is formed by a ray passing through a vertex $r'$ of \P and ending at a point $r$ on the base. (In Section~\ref{implementation} we will differentiate between \e{fixed-endpoint} sides for which $r=r'$ is an endpoint of the base and \e{rotating} rays which rotate around $r'$ if \p moves; here we treat both types of sides with a single formula, since fixed-endpoint sides may be viewed as a special case of rotating sides with $r=r'$.)

To write the formula for the area of the triangle $pqr$, we follow \cite[Appendix~A.1]{shop} and assume that the base is the x-axis and that both $p=(x,y)$ and $r'=(a,b)$ lie above the base ($y,b\ge0$); then the abscissa of $r$ is $x-y(x-a)/(y-b)$ (Fig.~\ref{V}, right). Let $q'$ be the vertex that defines the other side, $pq$, of $pqr$; to simplify the formulas, assume w.l.o.g.\ that $q'$ lies on the y-axis: $q'=(0,d)$. The abscissa of $q$ is then $x-yx/(y-d)$, and the area $y|rq|/2$ of the triangle $pqr$ is
\begin{equation}\label{area}|pqr|=\frac{y^2}2\left(\frac{x-a}{y-b}-\frac{x}{y-d}\right)\end{equation}

\begin{figure}\centering
\begin{minipage}[b]{.65\columnwidth}\centering\includegraphics{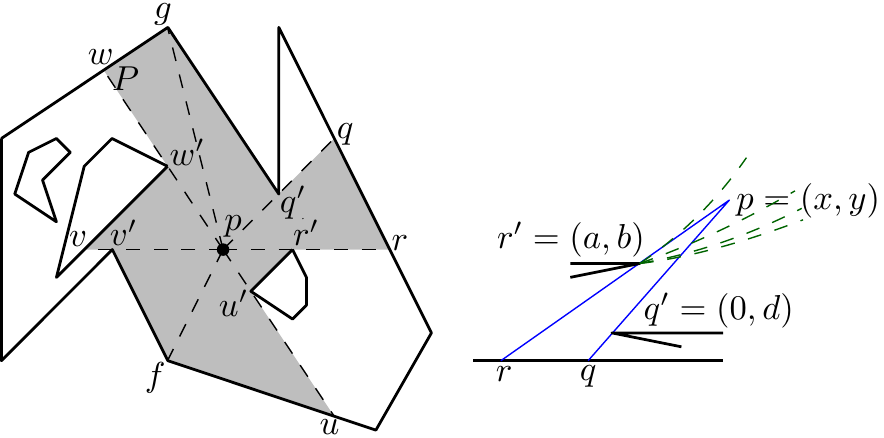}\caption{Left: Domain \P with 3 holes and a point $\p\in\P$; \V\p\ is shaded. Triangle $pqr$ has two rotating sides, triangles $puf,pvw',pwg$ have one fixed-endpoint and one rotating side; the other triangles have two fixed-endpoint sides. Right: $|pqr|=y|rq|/2$. Green dashed curves are level sets of $|pqr|$.}\label{V}
\end{minipage}
\hfil
\begin{minipage}[b]{.3\columnwidth}\centering\includegraphics{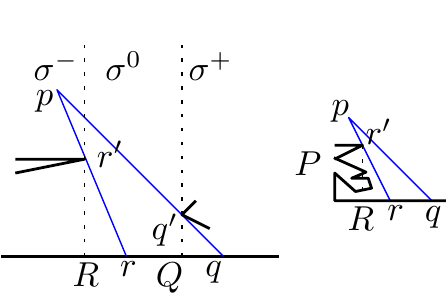}\caption{Left: For $\p\in\sigma^-$, $|r'Rr|$ is subtracted from $C$ while $|q'Qq|$ is added; for $\p\in\sigma^0$, both areas are added; for $\p\in\sigma^+$, $|r'Rr|$ is added while $|q'Qq|$ is subtracted. Right: $r'R$ is not fully inside~\P.}\label{refine}
\end{minipage}
\end{figure}

Next, to obtain a piecewise-constant (1+\eps)-approximation of the area |\V{$(x,y)$}| visible from point $(x,y)\in\P$, we use level sets of the area function (\ref{area}). For a given area $A$, the equality $|pqr|=A$ is attained along the curve $\gamma_A$
\begin{equation}\label{gamma}
x=\frac{2A/y^2+a/(y-b)}{1/(y-b)-1/(y-d)}.
\end{equation}%=\frac{\frac{2A}{y^2}-\frac{a}{y-b}}{\frac1{y-b}-\frac1{y-d}}
Consider a cell $\sigma$ of the visibility decomposition. We split $\sigma$ with the curves $\gamma_{A_i}$ for a set $\mathcal{A}=(A_1,\dots,A_i,\dots)$ of areas forming geometric progression with common ratio 1+\eps: $A_i=(1+\eps)A_{i-1}$. Let $S_i$ denote the set of points \p for which the area of the triangle $pqr$ is between $A_{i-1}$ and $A_i$ (that is, $S_i=\{\p\in\sigma:A_{i-1}<|pqr|\le A_i\}$ are the points between $\gamma_{A_{i-1}}$ and $\gamma_{A_{i}}$). We call $S_i$ a \e{curved sector} because in equation (\ref{gamma}), we have $\lim_{y\to b}x(y)=a$ for any $A$, i.e., all curves $\gamma_A$ have $r'=(a,b)$ as a common point. (We put a GeoGebra graphics to play with the level sets to see how they look at \url{https://www.geogebra.org/m/cvxvhfcf}.) We assign the same \e{weight} $A_i$ to all points in the curved sector; this way, for $i>1$ the weight of any point $\p\in S_i$ is within factor 1+\eps of the area of the triangle $pqr$:
\begin{equation}\label{w}|pqr|\le A_i\le(1+\eps)|pqr|\qquad\qquad\forall\p\in S_i,\forall i>1\end{equation}

For every cell $\sigma$ of the visibility decomposition, we overlay the level sets from each of the $O(n)$ triangles of \V\p\ for $\p\in\sigma$. We confine the level sets to the cell, i.e., for each curve $\gamma_A$ use only the intersection $\gamma_A\cap\sigma$. We call each cell of the overlay a \e{region} and set the weight of the region to the sum of the weights of the curved sectors whose intersection forms the region.

To bound the number of level sets used (i.e., to determine the first area $A_1$ in the geometric sequence $\cal A$ and the needed length of the sequence), assume that vertices of \P have integer coordinates and let \L denote the largest coordinate. (This model and its variants are common for WRP; in particular, the running times of known solutions for WRP \cite{wrp,inkulu,aleksandrov,anisotropic,querying,refinement} depend on \L.) Now, consider a triangulation $T$ of \P{} -- any point $\p\in\P$ lies inside a triangle $\tau$ of $T$ and sees all of the triangle; thus the area |\V\p| is at least the area of $\tau$. Since $\tau$ has integer coordinates, by Pick's Theorem \cite{pick} the area of the triangle is at least 1/2:\begin{equation}\label{pick}|\V\p|\ge1/2\end{equation}

We are now ready to prove that it suffices to have \begin{equation}\label{A1}A_1=\frac\eps{2n}\end{equation}Indeed, suppose \V\p\ consists of $K$ triangles of areas $\Delta_1,\dots,\Delta_K$ and let $A_1,\dots,A_K$ be the weights of the curved sectors that form the region to which \p belongs; the weight of the region is thus $w(p)=A_1+\dots+A_K$. Classify the triangles as ``small'' and ``large'', with the former having area at most $A_1$ (and thus having \p lie in the sector $S_1$) and the latter having area larger than $A_1$ (with \p in a sector $S_i$ for $i>1$); let $l=\{k:\Delta_k>A_1\}$ be the indices of the large triangles. By (\ref{w}), for every large triangle $k\in l$, $A_k\le(1+\eps)\Delta_k$. Since $K\le n$, we have
\begin{equation}\label{wapx}w(p)=\sum_{k\in l}A_k+\sum_{k\notin l}A_k\le(1+\eps)\sum_{k\in l}\Delta_k+nA_1\le(1+\eps)|\V\p|+\eps\frac12\le(1+2\eps)|\V\p|\end{equation}
where the last inequality is due to (\ref{pick}).
\begin{proposition}\label{prop}If WRP on $N$ regions with curved boundaries of constant algebraic complexity can be (1+\eps)-approximated in time $T(N,\frac1\eps)$, then a $(1+\eps)^2$-approximation to the minimum integral exposure path can be found in time $T(\frac{n^{10}}{\eps^2}\log^2(nL),\frac1\eps)$.\end{proposition}
\begin{proof}For an upper bound on the sector weight, note that obviously $\forall\p\in\P,\, |\V\p|\le\L^2$. Hence, the number of needed level sets is at most $\log_{1+\eps}(2nL^2)=O(\frac1\eps\log(n\L))$. The level sets are defined for each of the $O(n^3)$ triples $r',q',\bar{qr}$ where $r',q'$ are vertices and $\bar{qr}$ is the side of \P containing $qr$; thus overall there are $O(\frac{n^3}\eps\log(n\L))$ level set curves. Since each curve $\gamma_A$ has constant algebraic degree (cf.~(\ref{gamma})), any two curves intersect $O(1)$ times, so the complexity of the overlay of the level sets inside the cell $\sigma$ of the visibility decomposition is $O((\frac{n^3}\eps)^2\log^2(nL))$. Since there are $O(n^4)$ cells, our construction splits \P into $O(\frac{n^{10}}{\eps^2}\log^2(nL))$ regions of constant weight.

By (\ref{wapx}), region weights approximate the visibility area to within 1+\eps (use \eps:=\eps/2 to get rid of the factor 2 in front of \eps); hence finding a (1+\eps)-approximate solution to the WRP on our regions provides a $(1+\eps)^2$-approximation to the minimum integral exposure path. Formally, let $\pi^*$ be the minimum integral exposure path (the optimal solution to {\sc Integral Geometric Secluded Path}), let $\bar\pi$ be the minimum-weight path through our regions (the optimal solution to WRP) and let $\pi$ be the (1+\eps)-approximate solution to WRP; then
\begin{equation}\label{eint}\int_{\pi}|\V\p|\,\mathrm{d}p\,\le\,\int_{\pi}w(p)\,\mathrm{d}p\,\le\,(1+\eps)\int_{\bar\pi}w(p)\,\mathrm{d}p\,\le\,(1+\eps)\int_{\pi^*}w(p)\,\mathrm{d}p\,\le\,(1+\eps)^2\int_{\pi^*}|\V\p|\,\mathrm{d}p\end{equation}
where the first inequality is due to the left inequality of~(\ref{w}), the second is because $\pi$ approximates $\bar\pi$, the third is because $\bar\pi$ is optimal w.r.t.\ $w$, and the last one is due to the right inequality in~(\ref{w}).\end{proof}

\subsection{A detailed implementation}\label{implementation}Applicability of Proposition~\ref{prop} remains questionable due to absence of an algorithm for WRP with curved regions boundaries. In this section we present another, direct approach to reduce our problem to WRP on a \e{polygonal} subdivision. We refine the visibility decomposition (without affecting the asymptotic complexity) and recalculate the area functions so that they have \e{linear} levels. This way, the regions in the overlay of the level sets are convex, so existing WRP solutions can be applied directly.

Specifically, we differentiate between fixed-endpoint and rotating sides of the triangles into which \V\p\ is split: the former end at a vertex of \P while the latter rotate around a vertex if \p moves (see Fig.~\ref{V}, left). Triangles whose both sides are fixed-endpoint are easy to handle: (while the area of each individual triangle changes as \p moves,) the \e{total} area of all such triangles remains the same (moving \p just redistributes the area between the triangles, ``stealing'' from some and ``giving'' to others). We therefore call such triangles \e{fixed}.

Consider now a triangle $pqr$ whose both sides $pq,pr$ are rotating around vertices $q',r'$ resp.\ (this is the most general case: if one of the sides, say, $pq'$ is fixed, we can just assume $q=q'$); assume that $rq$ is horizontal (Fig.~\ref{refine}, left). We refine the visibility decomposition by extending the vertical segments through each of $r',q'$ maximally up and down; let $R,Q$ be the feet of the perpendiculars dropped from $r'$ and $q'$ resp.\ onto the supporting line of $pq$ (any of $r'R,q'Q$ may lie only partially inside \P, as in Fig.~\ref{refine}, right -- this is not an issue). Note that $|Rr'pq'Q|$ may be added to the fixed-triangles areas -- the total area of all fixed triangles plus the areas of the pentagons $Rr'pq'Q$ for all the triangles with \p as the apex does not depend on \p (while \p remains in the same cell). Denote this total area by $C$. The area $|\V\p|$ is obtained from $C$ by adding/subtracting the areas of the triangles $r'Rr$ for all vertices $r'$ on which a side of a triangle of \V\p\ rotates -- whether $|r'Rr|$ is added or subtracted depends on whether the triangle is in \V\p\ or not:
\begin{equation}\label{C}|\V\p|=C+\sum_{r'\in\oplus}\d-\sum_{r'\in\ominus}\d\end{equation}
where $\d=|r'Rr|$ and $\oplus$ (resp.\ $\ominus$) is the set of vertices whose triangles $r'Rr$ are visible (resp.\ invisible) from \p.

\begin{figure}\centering\includegraphics{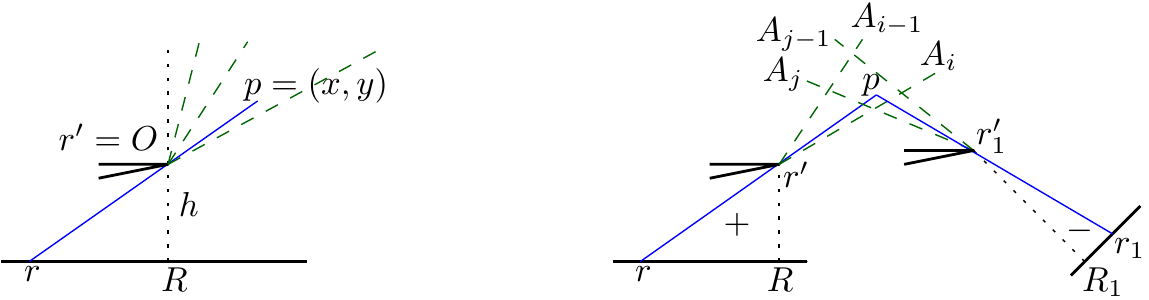}
\caption{Left: Green dashed are level sets of the area function \d (\ref{delta}). Right: $A_i$ contributes positively to $w(p)$ while $A_j$ comes with minus into $w(p)$, because for \p in this region, $r'\in\oplus$ ($r'Rr\in\V\p$) while $r_1'\in\ominus$ ($r_1'R'r_1'\notin\V\p$).}\label{rrr}\end{figure}

Assume that $r'$ is the origin $O$ and that the supporting line of $rR$ is the horizontal line $y=-h$, and let $p=(x,y)$ with $x\ge0$ (Fig.~\ref{rrr}, left). Then
\begin{equation}\label{delta}\d=\frac{h^2}2\frac xy\end{equation}
and a level set $\gamma_A=\{p=(x,y):\d=A\}$ of the function~(\ref{delta}) is a ray (emanating from the origin) of constant $x/y$: since the height $r'R$ of the triangle is fixed, \d is constant whenever $r$ is fixed. As in Section~\ref{generic}, we draw the rays for a set $\mathcal{A}=(A_1,\dots,A_i,\dots)$ of areas forming geometric progression with common ratio 1+\eps and assign the weight $A_i$ to all points in the sector $S_i=\{\p\in\sigma:A_{i-1}<\Delta_r\le A_i\}$ between $\gamma_{A_{i-1}}$ and $\gamma_{A_i}$ (we again use the weight $A_1=\eps/(2n)$ for points between $\gamma_0$ and $\gamma_{A_1}$).
Also as in Section~\ref{generic}, we define a \e{region} as a cell in the overlay of the rays emanating from the vertices $r'$ of \P. Finally, the weight $w(p)$ of any point \p in a region is determined by $C$ and the weights of the sectors forming the region: for a vertex $r'\in\oplus$ the weights of the sectors of $r'$ are added to regions weights; for a vertex $r'\in\ominus$, the weights are subtracted (Fig.~\ref{rrr}, right).

The fact that our subdivision into regions provides a (1+\eps)-approximation to |\V\p| can be argued similarly to Section~\ref{generic};
see Appendix~\ref{app:ptas} for the formulas:
\begin{restatable}{theorem}{tptas}\label{t:ptas}If WRP on $N$ regions can be (1+\eps)-approximated in time $T(N,\frac1\eps)$, then a $(1+\eps)^3$-approximation to the minimum integral exposure path can be found in time $T(\frac{n^{4}}{\eps}\log(nL),\frac1\eps)$.\end{restatable}

\section{On planar optimal satisfiability}\label{2sat}In this section we return to the (non-integral) {\sc Geometric Secluded Path} problem (Section~\ref{smin}) and elaborate on its connections to planar satisfiability, identifying, in particular, polynomially solvable and hard versions of planar minSAT and maxSAT. %In Section~\ref{sgraphs} we discuss connections of {\sc Geometric Secluded Path} to the graph version of the \sp problem.

For a SAT instance with variables $V$ and clauses $C$, the graph $G=(V\cup C,E)$ of the instance is the bipartite graph whose vertices are the variables and the clauses, and whose edges connect each variable to a clause whenever the variable or its negation appears in the clause. In a \e{planar} SAT, $G$ is planar. Planar SAT has been the staple starting point for hardness reduction in computational geometry. In many cases, hardness of geometric problems was proved using \e{restricted} hard versions of planar SAT, such as:\begin{description}\item[V-cycle SAT:]$G$ remains planar after adding a cycle through $V$ ($G$ is no longer bipartite)
\item[VC-cycle SAT:]$G$ remains planar after adding a cycle through $V\cup C$ (this version, as well as V-cycle SAT were defined already in the original paper on planar SAT \cite{planar})
\item[Separable SAT:]A further restriction of V-cycle SAT: for any variable $x$, the V-cycle separates clauses containing $x$ from the clauses containing $\overline x$; in other words, no variable $x$ has an $x$-containing clause and a $\overline x$-containing clause on the same side of the V-cycle (this version is from \cite[Lemma~1]{planar}, but has no name there; we take the name from \cite{tippenhauer})
\item[Monotone SAT:]In any clause, all variables are either non-negated or all variables are negated (this version is defined for general, not only for planar SAT)\end{description}
See \cite{erik,tippenhauer,pilz} for in-depth treatment of restricted planar SAT versions and their uses.

When proving hardness of {\sc Geometric Secluded Path} in Section~\ref{smin} (Theorem~\ref{reduction}) we spent considerable effort on dealing with crossings between variable--clause connectors. A natural question is why we did not reduce from planar minSAT. The answer is that to avoid crossings, our reduction should better start from separable minSAT (Fig.~\ref{sep}, left), so that for any variable $x$, the connections from literal $x$ reside on one side of the Christmas tree and the connections from $\overline x$ -- on its other side (otherwise, a connection from, say, $\overline x$ would cross the Christmas tree itself; Fig.~\ref{sep}, middle). However:\begin{theorem}\label{t:separable:min}Separable minSAT can be solved in polynomial time.\end{theorem}\begin{proof}Let $A$ be the clauses on one side of the variable chain and $B=C\setminus A$ -- the clauses on the other side. Construct the ``clause conflict'' graph $H$ \cite{ipl} whose vertices are the clauses and whose edges connect two clauses whenever one contains the negation of a literal in the other (Fig.~\ref{sep}, right). For any edge, at least one of the conflicting clauses will be satisfied in any truth assignment; thus, every edge in the graph will be incident to a satisfied clause. In particular, solving the minSAT is equivalent to finding minimum vertex cover (VC) in $H$. By the separability, for any variable $x$, all clauses with $x$ are in $A$ and all clauses with $\overline x$ are in $B$ (or vice versa); thus, any edge of $H$ connects a clause in $A$ with a clause in $B$, i.e., $H$ is bipartite, and the VC in it can be found in polynomial time.\end{proof}
\begin{figure}\centering\includegraphics{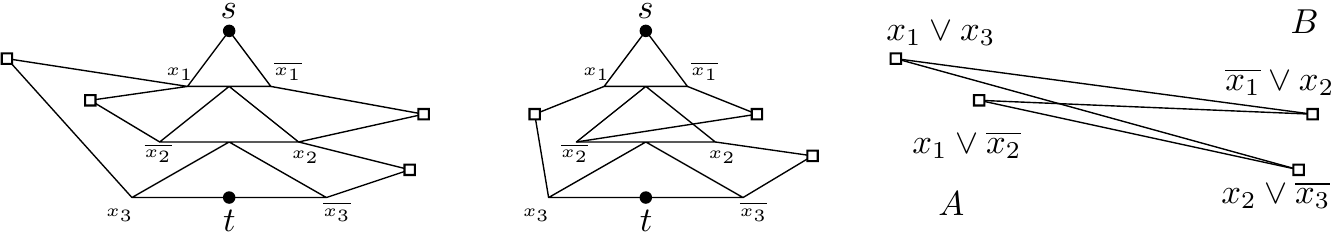}\caption{Left: Reduction from separable minSAT to {\sc Geometric Secluded Path} would have no crossings (note that some variables have their negations on one side of the Christmas tree, while others -- on the other; this is fine, since the definition of separable SAT requires separability \e{locally} for each variable; the separability does not have to be consistent across all variables). Middle: In non-separable minSAT, clause $\overline{x_1}\lor\overline{x_2}$ could be seen not only from \s-\t path via $\overline{x_2}$ but also from \s-\t path via $x_2$ due to the crossing with the Christmas tree. Right: the graph $H$ (which happens to be $K_{2,2}$) for the instance on the left.}\label{sep}\end{figure}
Note that the above proof does not use the planarity. In particular, monotone minSAT can be solved similarly: the clauses with all positive variables can form the set $A$ and the clauses with all negative variables -- set $B$ in the graph $H$ from the proof. We thus have:\begin{corollary}\label{c:monotone:min}Monotone minSAT (planar or not) can be solved in polynomial time.\end{corollary}
In Appendix~\ref{app:sat} we prove NP-hardness of V- and VC-cycle min2SAT, as well as hardness of all four versions of planar max2SAT (these do not have relation to secluded paths; we give the proofs just for completeness of our treatment of planar optSAT):\begin{restatable}{theorem}{toptsat}\label{t:opt2sat}The following planar versions of max2SAT are NP-hard: V-cycle, VC-cycle, monotone, separable. V- and VC-cycle min2SAT are NP-hard.\end{restatable}

\section{Conclusion}We studied minimum-exposure paths in polygonal domains. We showed that minimizing seen area is hard in polygons with (large number of) holes, while in polygons with a small number of holes the \s-\t path that sees least area can be found in polynomial time. We also gave a PTAS for finding an \s-\t path minimizing the integral of the seen area along the path. Finally, we discussed the connection between the geometric secluded paths and optimizing planar satisfiability, and identified hard and easy cases of planar optSAT (while the planar optSAT variants, which we proved hard, were not used in reductions in this paper, we hope that they may be useful in other settings). We conclude with some remarks on each of the problems studied.
\paragraph*{Minimizing seen area and Secluded paths in graphs}%\label{sgraphs}
Recall that in \sp (the original, graph problem) the goal is to find an \s-\t path adjacent to fewest vertices of the graph (vertices of the path itself are also counted as adjacent to the path). The problem was proved hard in \cite{chechik}. Our proof of hardness of {\sc Geometric Secluded Path} (Theorem~\ref{reduction}) gives an alternative proof of hardness of \sp in graphs: simply remove equalizers and leakage-blocking chambers from Fig.~\ref{reduction} (no need to care about midway intersections and all the other geometric technicalities) and add a large number of extra vertices adjacent to each clause vertex (Fig.~\ref{graph}, left).
While our proof is simpler than the ones in \cite{chechik}, it is less powerful because Chechik et al.~\cite{chechik} showed also hardness of approximation. In fact, the reduction in \cite{chechik}, shown here on Fig.~\ref{graph}, right, may also be seen as reduction from minSAT (in view of the connection between minSAT and VC in the clause conflict graph -- see proof of Theorem~\ref{t:separable:min}): the choices that the \s-\t path makes in the edges of the original graph $G$ may be seen as setting the truth values to the variables (similarly to how the path through our Christmas tree does it).
\begin{figure}\centering\includegraphics{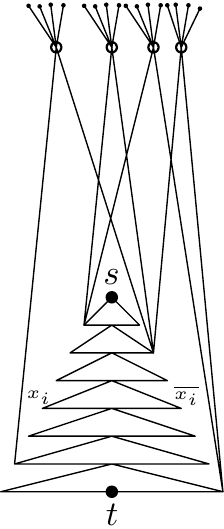}\hfil\includegraphics[height=5.3cm]{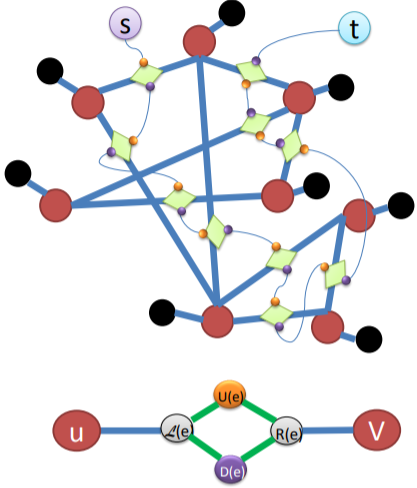}\caption{Left: Our reduction from min2SAT to \sp. To avoid high-degree vertices at the clauses (hollow), the \s-\t path will go via the Christmas tree, setting the variables; the number of seen (i.e., adjacent) clause vertices is the number of satisfied clauses. Right: The reduction from VC in a graph $G$ \cite[Fig.~3]{chechik}: the new graph $G'$ has new vertices \s and \t, and an \s-\t path (thin blue) crossing all edges (thick blue) is added to $G$, with every crossing (lightgreen rhombi) turned into a gadget (bottom) where the \s-\t path chooses which vertex of $G$ (red) the path will see; leaking into the original vertices of $G$ (red) is prevented in $G'$ by attaching high-weight vertices (black).}\label{graph}\end{figure}

A natural question, arising in view of the effort we spent dealing with the crossings in Section~\ref{smin} when proving hardness of {\sc Geometric Secluded Path} (Theorem~\ref{reduction}), is why we did not reduce from \sp in \e{planar} graphs. The answer is that we are not aware of a hardness result for the problem in planar graphs. Indeed, even though Chechik et al.'s hardness proof for \e{general} graphs (refer to Fig.~\ref{graph}, right) could reduce from VC in a \e{planar} graph $G$, in order to keep the planarity also in the resulting graph $G'$ (in which the secluded \s-\t path is sought), the added path (crossing all edges of $G$) must cross each edge exactly once, meaning that it is an Euler path in the planar dual of $G$, meaning that the dual has vertices of even degree only, meaning that $G$ has faces with even number of edges, meaning it has only even cycles, meaning it is bipartite, meaning VC is polynomial in it. (Strictly speaking, since we need only an Euler path through the edges, not Euler cycle, $G$ may have 2 odd faces -- we believe VC is still polynomial in such graphs).% E.g., if the odd cycles are triangles, the graph is perfect, implying that minimum VC in it can be found efficiently, as the complement of the max independent set. Also, if the two odd faces of a planar graph merge after removal of $k$ edges (which can be determined by finding the shortest path between the faces in the dual graph) then the VC can be solved in $O(2^k\mathrm{poly})$ time by trying all assignments of the endpoints of the $k$ edges into the VC -- this implies that if VC is hard in nearly-bipartite planar graphs, the two odd faces must be separated by long paths (in the dual).} Thus, likely, different techniques would be needed to prove hardness of \sp in planar graphs (assuming the problem is hard at all).

\paragraph*{The PTAS for integral seen area minimization}
Several remarks on the complexity of our solution:
\begin{itemize}
\item A faster algorithm for our problem could potentially be obtained by using a ``1D'' discretization of edges of the visibility decomposition (instead of creating a 2D ``grid'' of regions, as we do), as done in many algorithms for WRP (and related problems on minimizing path integral \cite{quality1,quality2}). Such a solution, however, would require knowing the optimal path connecting points on the boundary of the same cell of the decomposition. This, may be quite complicated, as it amounts to minimizing the integral of a function with $\Omega(n)$ terms, for which an analytical solution might not exist (though an approximation may be possible).
\item  An algorithm for WRP with regions whose boundaries are curves of constant algebraic degree could be interesting and would lead to a solution of our problem just using the generic scheme from Section~\ref{generic}. The biggest stumbling block for the design of such an algorithm may be the non-convexity of the regions, implying that a segment between two points on the boundary of a region is not guaranteed to stay inside the region. It may be possible that WRP techniques could be adapted to handle our regions from Section~\ref{generic} by approximating their boundaries with piecewise-linear functions (since we are looking only for a (1+\eps)-optimal path, the fineness of such piecewise-linear approximation would also be controlled by \eps).
\item Since our problem is an extension of WRP to the case of continuously changing weight, it may be tricky to establish hardness of the problem, as the complexity of WRP has remained open for many years (see \cite{algebraic} for a recent proof of algebraic complexity of WRP). Differently from 0/1 exposure (Theorem~\ref{simple}), even in simple polygons the shortest path does not necessarily minimize the integral exposure (Fig.~\ref{counterex}).
\end{itemize}
\begin{figure}\centering%\includegraphics{decrease}\hfill
\includegraphics[width=.9\columnwidth]{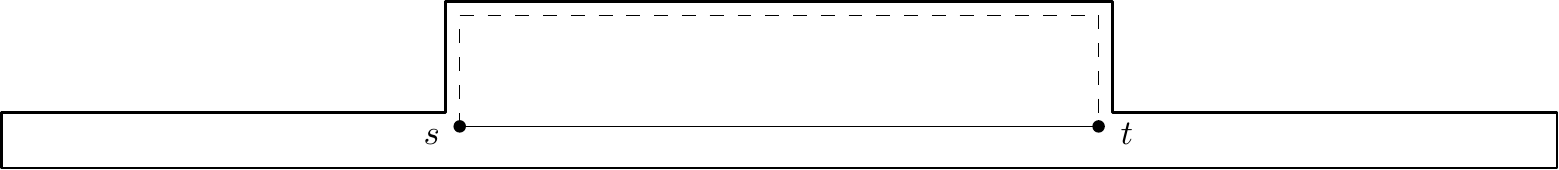}\caption{%Left: Red and blue edges of \P create red and blue level sets emanating from the same vertex; they could possibly be replaced by a single sequence of level sets providing the necessary approximations both for the red and the blue. Right:
Shortest \s-\t path (solid) sees the niches behind \s and \t for its whole length; stepping to the side (dashed path) decreases the integral exposure.}\label{counterex}\end{figure}

\paragraph*{Optimal 2-satisfiability}
Few observations on min2SAT and max2SAT:
\begin{itemize}\item Monotone minSAT is an example of the tractable class of submodular function minimization \cite{submodular}. 
\item Planar max2SAT has a PTAS \cite[Thm.~8.8]{SATbook}.
\item If in a separable max2SAT with VC cycle, the cycle also separates the variables at the clauses (i.e., if at each clause the connections from the two variables come from the different sides of the cycle), then the problem can be solved in polynomial time by reduction to separable min2SAT (Theorem~\ref{Kevin} in the Appendix~\ref{app:sat}).
\end{itemize}

\subparagraph*{Acknowledgements.}We thank Mike Paterson for raising the question of finding minimum-exposure paths, and the anonymous reviewers for the comments improving the presentation of the paper; we also acknowledge discussions with Irina Kostitsyna, Joe Mitchell and Topi Talvitie. Part of the work was done at the workshop on Distributed Geometric Algorithms held in the University of Bologna Centre at Bertinoro Aug 25-31, 2019. VP and LS are supported by the Swedish Transport Administration and the Swedish Research Council. 
\bibliographystyle{plainurl}\bibliography{document}\appendix

\begin{thebibliography}{10}

\bibitem{quality2}
Pankaj~K Agarwal, Kyle Fox, and Oren Salzman.
\newblock An efficient algorithm for computing high-quality paths amid
  polygonal obstacles.
\newblock {\em ACM Transactions on Algorithms (TALG)}, 14(4):46, 2018.

\bibitem{aleksandrov}
Lyudmil Aleksandrov, Anil Maheshwari, and J-R Sack.
\newblock Determining approximate shortest paths on weighted polyhedral
  surfaces.
\newblock {\em Journal of the ACM (JACM)}, 52(1):25--53, 2005.

\bibitem{qvm}
Esther~M Arkin, Alon Efrat, Christian Knauer, Joseph Mitchell, Valentin
  Polishchuk, G{\"u}nter Rote, Lena Schlipf, and Topi Talvitie.
\newblock Shortest path to a segment and quickest visibility queries.
\newblock {\em Journal of Computational Geometry}, 7(2):77--100, 2016.
\newblock Special issue on SoCG'15.

\bibitem{mobihoc}
Boris Aronov, Alon Efrat, Ming Li, Jie Gao, Joseph~SB Mitchell, Valentin
  Polishchuk, Boyang Wang, Hanyu Quan, and Jiaxin Ding.
\newblock Are friends of my friends too social? limitations of location privacy
  in a socially-connected world.
\newblock In {\em Proceedings of the Eighteenth ACM International Symposium on
  Mobile Ad Hoc Networking and Computing}, pages 280--289. ACM, 2018.

\bibitem{decomposition}
Boris Aronov, Leonidas~J Guibas, Marek Teichmann, and Li~Zhang.
\newblock Visibility queries and maintenance in simple polygons.
\newblock {\em Discrete \& Computational Geometry}, 27(4):461--483, 2002.

\bibitem{bengt}
Svante Carlsson, H{\aa}kan Jonsson, and Bengt Nilsson.
\newblock Finding the shortest watchman route in a simple polygon.
\newblock {\em Discrete {\&} Computational Geometry}, 22(3):377--402, 1999.

\bibitem{chechik}
Shiri Chechik, Matthew~P Johnson, Merav Parter, and David Peleg.
\newblock Secluded connectivity problems.
\newblock {\em Algorithmica}, 79(3):708--741, 2017.

\bibitem{refinement}
Siu-Wing Cheng, Jiongxin Jin, and Antoine Vigneron.
\newblock Triangulation refinement and approximate shortest paths in weighted
  regions.
\newblock In {\em Proceedings of the twenty-sixth annual ACM-SIAM symposium on
  Discrete algorithms}, pages 1626--1640. SIAM, 2014.

\bibitem{anisotropic}
Siu-Wing Cheng, Hyeon-Suk Na, Antoine Vigneron, and Yajun Wang.
\newblock Approximate shortest paths in anisotropic regions.
\newblock {\em SIAM Journal on Computing}, 38(3):802--824, 2008.

\bibitem{querying}
Siu-Wing Cheng, Hyeon-Suk Na, Antoine Vigneron, and Yajun Wang.
\newblock Querying approximate shortest paths in anisotropic regions.
\newblock {\em SIAM Journal on Computing}, 39(5):1888--1918, 2010.

\bibitem{shop}
Otfried Cheong, Alon Efrat, and Sariel Har-Peled.
\newblock Finding a guard that sees most and a shop that sells most.
\newblock {\em Discrete \& Computational Geometry}, 37(4):545--563, 2007.

\bibitem{SATbook}
Nadia Creignou, Sanjeev Khanna, and Madhu Sudan.
\newblock {\em Complexity classifications of boolean constraint satisfaction
  problems}, volume~7.
\newblock SIAM, 2001.

\bibitem{algebraic}
Jean-Lou De~Carufel, Carsten Grimm, Anil Maheshwari, Megan Owen, and Michiel
  Smid.
\newblock A note on the unsolvability of the weighted region shortest path
  problem.
\newblock {\em Computational Geometry}, 47(7):724--727, 2014.

\bibitem{erik}
Erik Demaine.
\newblock Algorithmic lower bounds: Fun with hardness proofs.
\newblock MIT OCW.

\bibitem{sensor}
Hristo~N Djidjev.
\newblock Efficient computation of minimum exposure paths in a sensor network
  field.
\newblock In {\em International Conference on Distributed Computing in Sensor
  Systems}, pages 295--308. Springer, 2007.

\bibitem{DYER1986174}
M.E Dyer and A.M Frieze.
\newblock Planar {3DM} is np-complete.
\newblock {\em Journal of Algorithms}, 7(2):174 -- 184, 1986.
\newblock URL:
  \url{http://www.sciencedirect.com/science/article/pii/0196677486900027},
  \href {http://dx.doi.org/https://doi.org/10.1016/0196-6774(86)90002-7}
  {\path{doi:https://doi.org/10.1016/0196-6774(86)90002-7}}.

\bibitem{sensor3}
Hao Feng, Lei Luo, Yong Wang, Miao Ye, and Rongsheng Dong.
\newblock A novel minimal exposure path problem in wireless sensor networks and
  its solution algorithm.
\newblock {\em International Journal of Distributed Sensor Networks},
  12(8):1550147716664245, 2016.

\bibitem{kulikov}
Fedor~V Fomin, Petr~A Golovach, Nikolay Karpov, and Alexander~S Kulikov.
\newblock Parameterized complexity of secluded connectivity problems.
\newblock {\em Theory of Computing Systems}, 61(3):795--819, 2017.

\bibitem{01socg}
Laxmi Gewali, Alex~C. Meng, Joseph S.~B. Mitchell, and Simeon~C. Ntafos.
\newblock Path planning in 0/1/infinity weighted regions with applications.
\newblock In Herbert Edelsbrunner, editor, {\em Proceedings of the Fourth
  Annual Symposium on Computational Geometry, Urbana-Champaign, IL, USA, June
  6-8, 1988}, pages 266--278. {ACM}, 1988.

\bibitem{ghoshBook}
Subir Ghosh.
\newblock {\em Visibility Algorithms in the Plane}.
\newblock Cambridge University Press, New York, NY, USA, 2007.

\bibitem{handbook}
J.E. Goodman and J.~O'Rourke, editors.
\newblock {\em Handbook of Discrete and Computational Geometry}.
\newblock Discrete Mathematics and Its Applications. Taylor \& Francis, 2nd
  edition, 2004.

\bibitem{pick}
Branko Gr{\"u}nbaum and Geoffrey~C Shephard.
\newblock Pick's theorem.
\newblock {\em The American Mathematical Monthly}, 100(2):150--161, 1993.

\bibitem{joe}
Leonidas~J Guibas, John~E Hershberger, Joseph~SB Mitchell, and Jack~Scott
  Snoeyink.
\newblock Approximating polygons and subdivisions with minimum-link paths.
\newblock {\em International Journal of Computational Geometry \&
  Applications}, 3(04):383--415, 1993.

\bibitem{homot}
John Hershberger and Jack Snoeyink.
\newblock Computing minimum length paths of a given homotopy class.
\newblock {\em Computational geometry}, 4(2):63--97, 1994.

\bibitem{submodular}
Dorit~S Hochbaum.
\newblock Complexity and approximations for submodular minimization problems on
  two variables per inequality constraints.
\newblock {\em Discrete Applied Mathematics}, 250:252--261, 2018.

\bibitem{inkulu}
Rajasekhar Inkulu and Sanjiv Kapoor.
\newblock A polynomial time algorithm for finding an approximate shortest path
  amid weighted regions.
\newblock {\em Preprint}, 2015.

\bibitem{minsatOrig}
Rajeev Kohli, Ramesh Krishnamurti, and Prakash Mirchandani.
\newblock The minimum satisfiability problem.
\newblock {\em SIAM Journal on Discrete Mathematics}, 7(2):275--283, 1994.

\bibitem{minlink3d}
Irina Kostitsyna, Maarten L{\"{o}}ffler, Valentin Polishchuk, and Frank Staals.
\newblock On the complexity of minimum-link path problems.
\newblock {\em JoCG}, 8(2):80--108, 2017.
\newblock Special Issue on SoCG'16.
\newblock URL: \url{http://jocg.org/index.php/jocg/article/view/328}.

\bibitem{occluded1}
Niel Lebeck, Thomas M{\o}lhave, and Pankaj~K Agarwal.
\newblock Computing highly occluded paths on a terrain.
\newblock In {\em Proceedings of the 21st ACM SIGSPATIAL International
  Conference on Advances in Geographic Information Systems}, pages 14--23. ACM,
  2013.

\bibitem{occluded2}
Niel Lebeck, Thomas M{\o}lhave, and Pankaj~K Agarwal.
\newblock Computing highly occluded paths using a sparse network.
\newblock In {\em Proceedings of the 22nd ACM SIGSPATIAL International
  Conference on Advances in Geographic Information Systems}, pages 3--12. ACM,
  2014.

\bibitem{planar}
David Lichtenstein.
\newblock Planar formulae and their uses.
\newblock {\em SIAM journal on computing}, 11(2):329--343, 1982.

\bibitem{luckow}
Max-Jonathan Luckow and Till Fluschnik.
\newblock On the computational complexity of length-and
  neighborhood-constrained path problems.
\newblock {\em arXiv preprint arXiv:1808.02359}, 2018.

\bibitem{ipl}
Madhav~V Marathe and SS~Ravi.
\newblock On approximation algorithms for the minimum satisfiability problem.
\newblock {\em Information Processing Letters}, 58(1):23--29, 1996.

\bibitem{mrw}
J.~Mitchell, G.~Rote, and G.~Woeginger.
\newblock {Minimum-link paths among obstacles}.
\newblock {\em Alg-ca'92}, 8(1):431--459, 1992.

\bibitem{minlink}
Joseph Mitchell, Valentin Polishchuk, and Mikko Sysikaski.
\newblock Minimum-link paths revisited.
\newblock {\em CGTA}, 47(6):651--667, 2014.

\bibitem{wrp}
Joseph~SB Mitchell and Christos~H Papadimitriou.
\newblock The weighted region problem: finding shortest paths through a
  weighted planar subdivision.
\newblock {\em Journal of the ACM (JACM)}, 38(1):18--73, 1991.

\bibitem{socg07}
Joseph~SB Mitchell and Valentin Polishchuk.
\newblock Thick non-crossing paths and minimum-cost flows in polygonal domains.
\newblock In {\em Proceedings 23rd ACM Symposium on Computational Geometry},
  pages 56--65, 2007.

\bibitem{agtBook}
Joseph O'Rourke.
\newblock {\em Art Gallery Theorems and Algorithms}.
\newblock The International Series of Monographs on Computer Science. Oxford
  University Press, New York, NY, 1987.

\bibitem{papadimitriou}
Christos~H Papadimitriou.
\newblock An algorithm for shortest-path motion in three dimensions.
\newblock {\em Information Processing Letters}, 20(5):259--263, 1985.

\bibitem{pilz}
Alexander Pilz.
\newblock Planar 3-sat with a clause/variable cycle.
\newblock {\em arXiv preprint arXiv:1710.07476}, 2017.

\bibitem{gender}
Valentin Polishchuk and Leonid Sedov.
\newblock Gender-aware facility location in multi-gender world.
\newblock In {\em LIPIcs-Leibniz International Proceedings in Informatics},
  volume 100. Schloss Dagstuhl-Leibniz-Zentrum fuer Informatik, 2018.

\bibitem{sensor2}
Yuning Song, Liang Liu, Huadong Ma, Athanasios~V Vasilakos, et~al.
\newblock A biology-based algorithm to minimal exposure problem of wireless
  sensor networks.
\newblock {\em IEEE Trans. Network and Service Management}, 11(3):417--430,
  2014.

\bibitem{suri}
Subhash Suri.
\newblock A linear-time algorithm for minimum link paths inside a simple
  polygon.
\newblock {\em Computer Vision, Graphics and Image Processing}, 35(1):99--110,
  1986.

\bibitem{tippenhauer}
Simon Tippenhauer and Wolfgang Muzler.
\newblock On planar 3-sat and its variants.
\newblock {\em Fachbereich Mathematik und Informatik der Freien Universitat
  Berlin}, 2016.

\bibitem{handbook3}
Csaba~D Toth, Joseph O'Rourke, and Jacob~E Goodman.
\newblock {\em Handbook of discrete and computational geometry}.
\newblock Chapman and Hall/CRC, 2017.

\bibitem{novosibirsk}
Ren{\'{e}} van Bevern, Till Fluschnik, and Oxana~Yu. Tsidulko.
\newblock Parameterized algorithms and data reduction for safe convoy routing.
\newblock In {\em 18th Workshop on Algorithmic Approaches for Transportation
  Modelling, Optimization, and Systems, {ATMOS} 2018, August 23-24, 2018,
  Helsinki, Finland}, pages 10:1--10:19, 2018.

\bibitem{sensor1}
Giacomino Veltri, Qingfeng Huang, Gang Qu, and Miodrag Potkonjak.
\newblock Minimal and maximal exposure path algorithms for wireless embedded
  sensor networks.
\newblock In {\em Proceedings of the 1st international conference on Embedded
  networked sensor systems}, pages 40--50. ACM, 2003.

\bibitem{Haitao}
Haitao Wang.
\newblock Quickest visibility queries in polygonal domains.
\newblock In Boris Aronov and Matya Katz, editors, {\em 33rd International
  Symposium on Computational Geometry, SoCG 2017, July 4-7, 2017, Brisbane,
  Australia}, volume~77 of {\em LIPIcs}, pages 61:1--61:16. Schloss Dagstuhl -
  Leibniz-Zentrum fuer Informatik, 2017.

\bibitem{quality1}
Ron Wein, Jur Van Den~Berg, and Dan Halperin.
\newblock Planning high-quality paths and corridors amidst obstacles.
\newblock {\em The International Journal of Robotics Research},
  27(11-12):1213--1231, 2008.

\bibitem{jurWAFR06}
Ron Wein, Jur Van Den~Berg, and Dan Halperin.
\newblock Planning near-optimal corridors amidst obstacles.
\newblock In {\em Algorithmic Foundation of Robotics VII}, pages 491--506.
  Springer, 2008.

\bibitem{vv}
Ron Wein, Jur~P Van~den Berg, and Dan Halperin.
\newblock The visibility--voronoi complex and its applications.
\newblock {\em Computational Geometry}, 36(1):66--87, 2007.

\end{thebibliography}

\section{Polygons with small number of holes}\begin{lemma}\label{homotopy}Let $ab$ be a segment in a polygonal domain \P, and let $\pi$ be an $a\textrm-b$ path homotopically equivalent to $ab$. If a point $\p\in\P$ is seen from $ab$, then it is also seen from $\pi$.\end{lemma}
\begin{proof}Assume that $\pi$ does not intersect $ab$ other than at $a,b$ -- this assumption is w.t.o.g.\ since we may separately consider each subpath of $\pi$ between consecutive intersection points with $ab$. Let $c\in ab$ be a point which sees \p (Fig.~\ref{fig:homot}). Extend the segment $cp$ maximally in both directions, until it hits the holes $H_1,H_2$ (it may be that $H_!=H_2$ and that any of $H_1,H_2$ is the outer polygon of \P); let $\bar{cp}$ be the extended segment. If $\pi$ does not intersect $\bar{cp}$, then at least one of the holes $H_1,H_2$ is inside the closed loop formed by $ab$ and $\pi$, implying that $\pi$ is not homotopically equivalent to $ab$ -- a contradiction. Thus, $\pi$ intersects $\bar{cp}$ and sees \p at the point of intersection.\end{proof}
\begin{figure}\centering\includegraphics{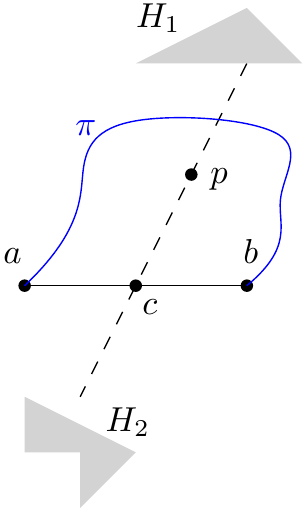}\caption{Both $H_1$ and $H_2$ must be outside the loop formed by $\pi$ and $ab$.}\label{fig:homot}\end{figure}
It follows that in a polygon with small number of holes, the secluded path may be found by scrolling through all homotopy types of simple \s-\t paths (e.g., by guessing the order in which the locally shortest path touches the holes \cite[Section~4]{socg07}), finding the shortest path of each homotopy type (e.g., using \cite{homot}) and choosing the locally shortest path that sees least (clearly, the area seen by a given path can be calculated in polynomial time).

\section{Formalities omitted from Section~\ref{implementation}}\label{app:ptas}

For a point $p\in P$, split the areas \d from (\ref{C}) into large (larger than $A_1$) and small (the others). By construction, the area \d of every large triangle $r'Rr$ is (1+\eps)-approximated by the weight $A_i$ of the sector $S_i$ $(i>1)$ to which \p belongs
\begin{equation}\label{large}\d\le A_i\le(1+\eps)\d\end{equation}while for small \d
\begin{equation}\label{small}0\le\d\le A_1\end{equation}
In particular, since there are at most $n$ small triangles, their total area is at most $nA_1=\eps/2\le\eps|\V\p|$ (cf.~(\ref{pick})); thus replacing \d with $A_1$ for each small triangle changes |\V\p| by at most an additive \eps|\V\p|.

Formally, let $\oplus_1\subseteq\oplus$ (resp.\ $\ominus_1\subseteq\ominus$) denote the vertices in $\oplus$ (resp.\ $\ominus$) for which \d is small. We can expand (\ref{C}) into \begin{equation}\label{expand}|\V\p|=C+\sum_{r'\in\oplus\setminus\oplus_1}\d+\sum_{r'\in\oplus_1}\d-\sum_{r'\in\ominus\setminus\ominus_1}\d-\sum_{r'\in\ominus_1}\d\end{equation}
The weight $w(p)$ is obtained by replacing \d in each summand with the area $A_i$ of the sector $S_i\ni p$:
\begin{equation}\label{ew}w(p)=C+\sum_{\oplus\setminus\oplus_1}A_i+\sum_{\oplus_1}A_1-\sum_{\ominus\setminus\ominus_1}A_i-\sum_{\ominus_1}A_1\end{equation} Using~(\ref{large}) and (\ref{small}), we get
\begin{equation}\label{expand-}w(p)\le C+(1+\eps)\sum_{r'\in\oplus\setminus\oplus_1}\d+\eps|\V\p|-\sum_{r'\in\ominus\setminus\ominus_1}\d-\sum_{r'\in\ominus_1}\d\end{equation}
Observe that even if the triangles from $\oplus$ are removed from \V\p, the point \p still sees a non-negative area, i.e., \[C-\sum_{r'\in\ominus\setminus\ominus_1}\d-\sum_{r'\in\ominus_1}\d\ge0\] which means that \[C-\sum_{r'\in\ominus\setminus\ominus_1}\d-\sum_{r'\in\ominus_1}\d\qquad\le\qquad(1+\eps)\left(C-\sum_{r'\in\ominus\setminus\ominus_1}\d-\sum_{r'\in\ominus_1}\d\right)\]
Substituting this into~(\ref{expand-}) and comparing with~(\ref{expand}), we get
\[w(p)\le(1+2\eps)|\V\p|\]

On the other hand, from (\ref{large})-(\ref{ew}),
\[w(p)\ge C+\sum_{\oplus\setminus\oplus_1}\d+\sum_{\oplus_1}\d-(1+\eps)\sum_{\ominus\setminus\ominus_1}\d-\sum_{\ominus_1}\frac{\eps|\V\p|}n\ge\]
\[\ge C+\sum_{\oplus\setminus\oplus_1}\d+\sum_{\oplus_1}\d-\sum_{\ominus\setminus\ominus_1}\d-\eps\sum_{\ominus\setminus\ominus_1}\d-\eps|\V\p|\ge\]
\[\ge C+\sum_{\oplus\setminus\oplus_1}\d+\sum_{\oplus_1}\d-\sum_{\ominus\setminus\ominus_1}\d-\sum_{\ominus_1}\d-\eps\sum_{\ominus\setminus\ominus_1}\d-\eps|\V\p|\ge\]
\[\ge|\V\p|-\eps|\V\p|-\eps|\V\p|=(1-2\eps)|\V\p|\]
Similarly to (\ref{eint}), the above proves the approximation ratio:
\[\int_{\pi}w(p)\,\mathrm{d}p\quad\le\quad\frac{(1+2\eps)(1+\eps)}{1-2\eps}\int_{\pi^*}|\V\p|\,\mathrm{d}p\]
where $\pi$ is the (1+\eps)-approximate path through our weighted regions and $\pi^*$ is the path with minimum integral exposure.

Overall, in comparison with Section~\ref{generic}, the regions in our WRP now have straightline boundaries, so a standard WRP algorithm can be applied. Also, the number of regions is decreased, because a sequence $\cal A$ of the level sets is defined by a \e{pair} $r',\bar{r}$ where $r'$ is a vertex of \P and $\bar r$ is the side of \P containing $r$; thus overall there are $O(\frac{n^2}\eps\log(nL))$ level sets. The level sets are overlaid with the $O(n^2)$ extensions of the visibility graph edges, defining the visibility decomposition (we let the level sets straddle through the cells of the visibility decomposition since the level sets are the same irrespective of the cell; the only term in formula~(\ref{C}) for |\V\p| that changes from cell to cell is $C$). Thus overall there are $O(\frac{n^4}{\eps}\log(nL))$ regions and similarly to Proposition~\ref{prop} we have:
\tptas*

\section{Mathematica listing}\label{listing}\includepdf[offset=150 -150]{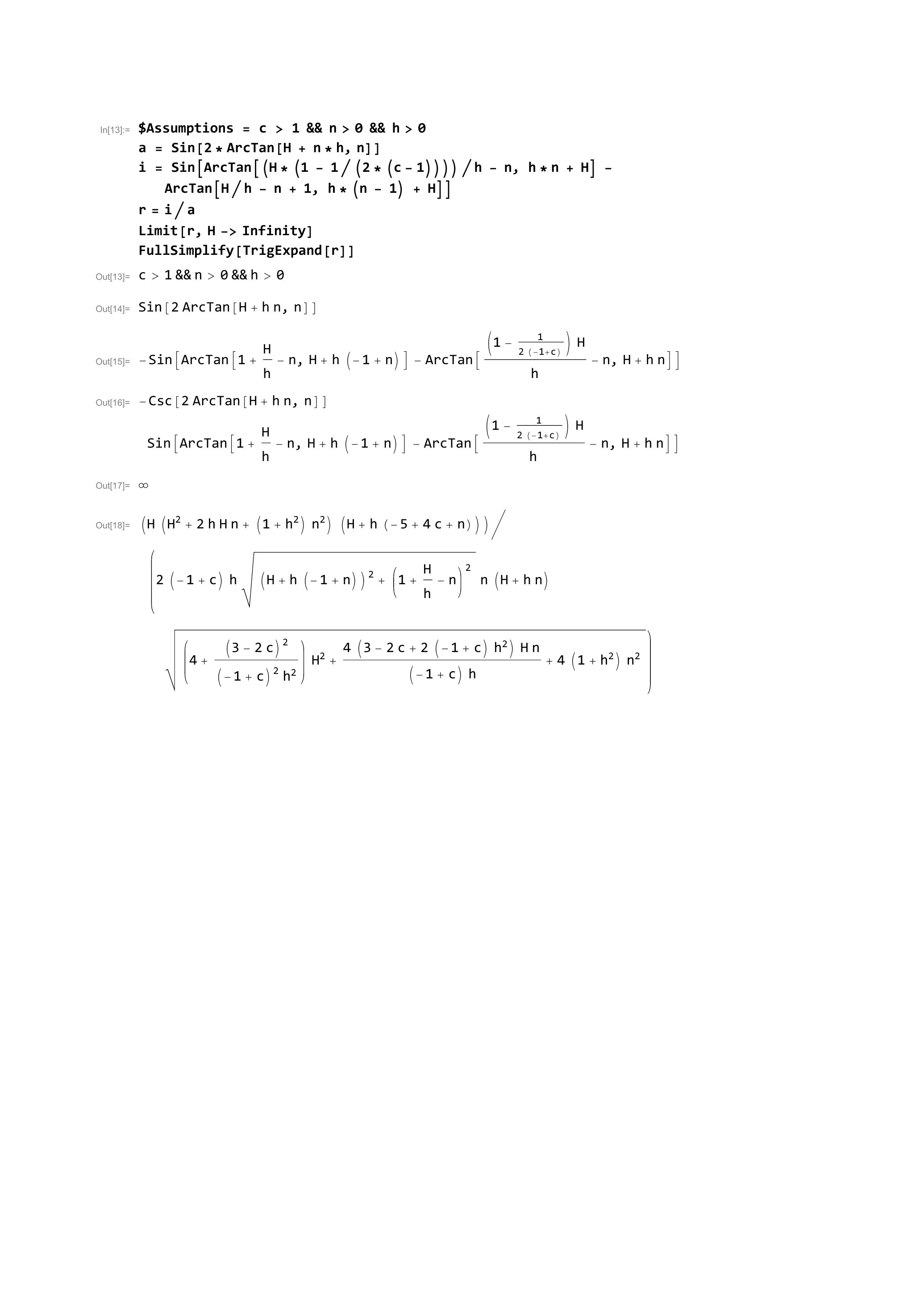}

\section{Hard planar versions of min2SAT and max2SAT}\label{app:sat}We prove Theorem~\ref{t:opt2sat}, restated also here:\toptsat*
\begin{proof}Planar min2SAT (without V- or VC-cycle) was proved hard by Guibas et al.\ in \cite[Theorem~3.2]{joe} using a reduction from planar 3SAT. We did not see how to add the cycles on top of Guibas et al.'s gadgets, and therefore present a different reduction. Some clauses in our 2SAT instances will have 1 literal (not 2) -- we do not differentiate between versions of 2SAT with \e{exactly} 2 literals per clause and \e{at most} 2 literals per clause (refer to \cite{erik,tippenhauer} for discussion of the differences between the versions, definitions, hardness proofs and uses of the many versions of satisfiability). We also do not show 1-literal clauses on pictures of our gadgets below, since it is straightforward to add them and extend the cycles to run through them.

To prove hardness of max2SAT with V-cycle, we reduce from V-cycle 1-in-3SAT (find truth assignment to satisfy \e{exactly} one literal in each clause) shown hard by Dyer and Frieze \cite{DYER1986174}.\footnote{Interestingly, Dyer and Frieze reported that they did not need the cycles for their purposes, but a referee insisted that others might need it later.} We replace every clause $a\lor b\lor c$ with 9 clauses: $\lnot a, \, \lnot b, \, \lnot c, \, \lnot a\lor\lnot b, \, \lnot b\lor \lnot c, \, \lnot a\lor\lnot c, \, a\lor b, \, b\lor c, \,a\lor c$. If none, 2 or 3 of $a,b,c$ are true, then 6 new clauses are satisfied, if 1 is true -- 7 are satisfied, and if all of $a,b,c$ are true -- 3 new clauses are satisfied. That is, $7|C|$ clauses are satisfied in the max2SAT instance iff all $|C|$ clauses are satisfied in the 1-in-3SAT. Figure~\ref{1in3}, left shows that if the 1-in-3SAT instance is planar, then the max2SAT instance is also planar; since the two instances have the same variables, the V-cycle in 1-in-3SAT is inherited by the max2SAT.\begin{figure}\centering\includegraphics[width=.3\columnwidth,page=2]{1in3}\hfil\includegraphics[width=.4\columnwidth,page=1]{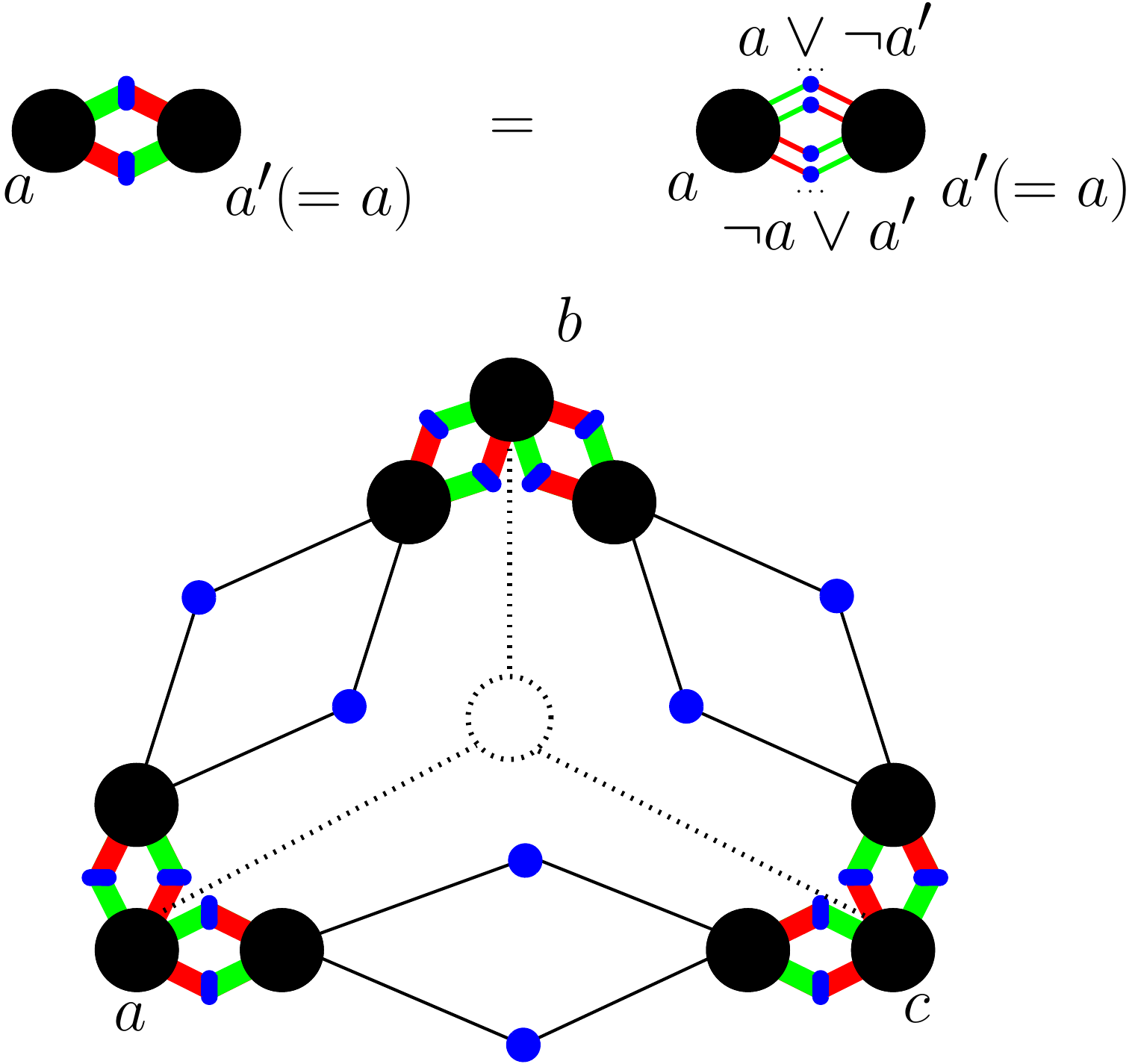}\caption{Left: Variables are large black disks and the six 2-literal clauses of max2SAT are blue circles; the 1-in-3SAT clause and connection to it are dotted. The V-cycle is blue. Right: Each thick red-green rhombus is the ``truth propagator'', consisting of the $2N$ clauses. Green are positive connections, red are negative.}\label{1in3}\end{figure}

To prove hardness of max2SAT with VC-cycle, we reduce from VC-cycle 1-in-3SAT, also shown hard in \cite{DYER1986174}. We start from the gadget we used to prove hardness of V-cycle max2SAT (Fig.~\ref{1in3}, left) and extend it by adding 2 copies of each of $a,b,c$ (Fig.~\ref{1in3}, right). The copy $a'$ of $a$ forms $N$ clauses $a\lor \lnot a'$ and $N$ clauses $\lnot a\lor a'$; the number $N$ is chosen so large that irrespectively of the rest of the instance, the optimal truth assignment would rather set $a=a'$ (which will satisfy all $2N$ clauses) than set $a=\lnot a'$ (which will satisfy only $N$). The same is done with the remaining 5 copies (the other copy of $a$, 2 copies of $b$ and 2 copies of $c$). Now, all $|C|$ clauses can be satisfied in 1-in-3SAT iff $7|C|+12N|C|$ clauses can be satisfied in the max2SAT. Finally, note that the VC-cycle in the 1-in-3SAT instance can go through the clause either as in Fig.~\ref{VC}, top left (not separating the variables) or as in Fig.~\ref{VC}, bottom left (separating the variables); Fig.~\ref{VC}, right shows how to run the cycle through the new variables and clauses of max2SAT in both cases.\begin{figure}\centering\includegraphics[width=.9\columnwidth,page=1]{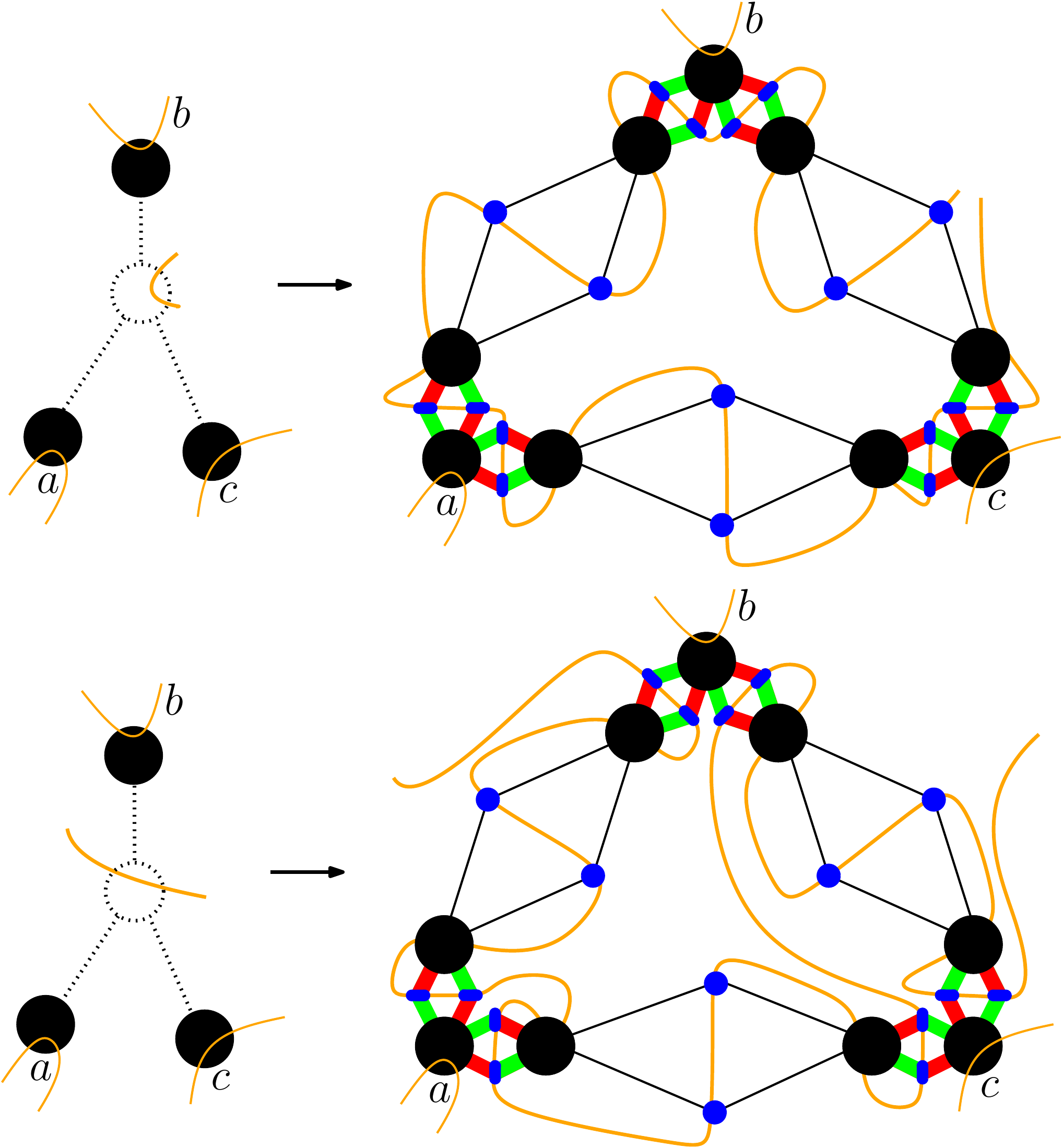}\caption{The VC-cycles are orange. Left: 1-in-3SAT. Right: max2SAT.}\label{VC}\end{figure}

Next, we prove hardness of planar monotone max2SAT by reduction from planar max2SAT, resolving non-monotone clauses one-by-one as shown in Figure~\ref{monotonicity_gadget}: in each non-monotone clause $x\lor\lnot y$ we select one of the variables, say $y$, and split it into three variables $y, z, t$ so that $y = \lnot z = t$ (the same trick as above, of introducing a large number $N$ of ``parallel'' monochromatic clauses $a\lor b$ and $\lnot a\lor\lnot b$ as the negator, is used to enforce $y=\lnot z=t$ in the optimal truth assignment), and replace $x\lor\lnot y$ with $x\lor z$ ($t$ is used to pick up all the clauses on the other side of the fixed, monotonized connection).\begin{figure}\centering\includegraphics[width=0.9\textwidth,page=3]{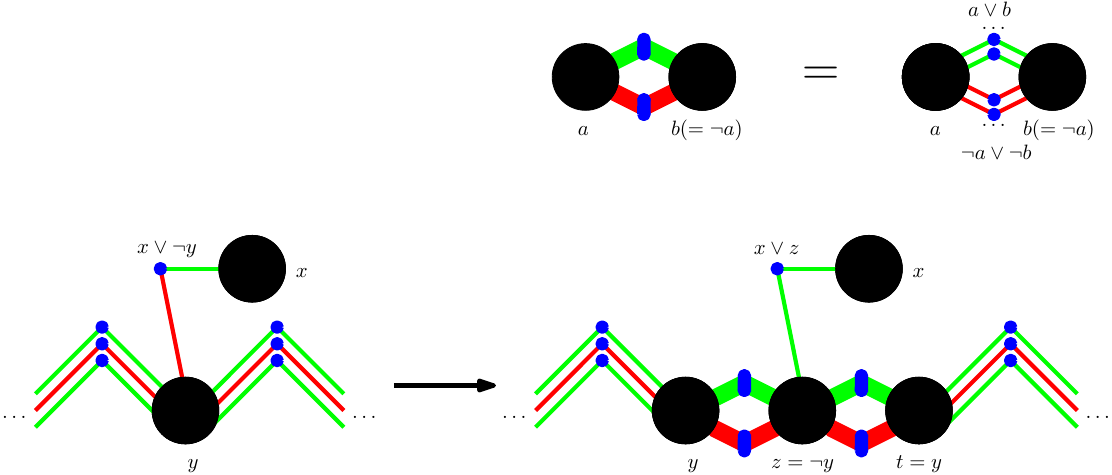}\caption{The monotonicity gadget. We make the clauses monochromatic one by one. As in the truth propagator in Fig.~\ref{1in3}, every thick edge in the negator is a large number of identical connections.}\label{monotonicity_gadget}\end{figure}
We prove hardness of separable max2SAT by reduction from planar max2SAT with V-cycle, using the same idea as in proving hardness of monotone planar min2SAT: for any variable $y$ that has both positive and negative connections on both sides of the V-cycle, we select one side to be ``positive'' and the other to be ``negative'', and if there are negative connections on the positive side, we split $y$ into three variables $y = \lnot z = t$ and fix the connection as shown in Figure~\ref{separability_gadget}.\begin{figure}\centering\includegraphics[width=0.9\textwidth,page=4]{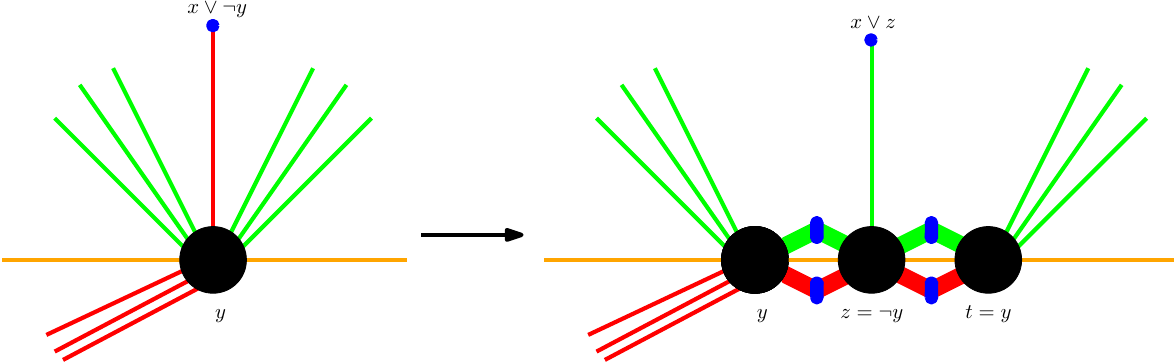}\caption{The separability gadget; V-cycle is orange.}\label{separability_gadget}\end{figure}

Finally, we show hardness of V- and VC-cycle min2SAT.
The hardness of (non-planar) min2SAT was originally shown by Kohli et al.~\cite{minsatOrig} by the following reduction from (non-planar) max2SAT: replace every clause $x\lor y$ with two clauses $\lnot x\lor w,\lnot y\lor w$, where $w$ is a new variable; as is easy to see, in the new instance it is possible to satisfy at most $2|C|-k$ clauses iff it was possible to satisfy at least $k$ clauses in the original instance ($|C|$ stands for the number of clauses in the original max2SAT instance). The reduction is easy to do in a planarity-preserving way; moreover, if we start from VC-cycle max2SAT, its VC-cycle becomes the V-cycle in min2SAT (Fig.~\ref{max2min_gadget}) -- this proves hardness of V-cycle min2SAT. To show hardness of VC-cycle min2SAT, we use the same reduction from V-cycle max2SAT, but look more closely at how the cycle goes through the clause; for each of the 3 different ways (the cycle is on one side of the clause, the other side, or crosses through the clause), we extend the cycle in the min2SAT so that it grabs also the new clauses (Fig.~\ref{max2minVC_gadget}).
\begin{figure}\centering\includegraphics[width=0.9\textwidth,page=2]{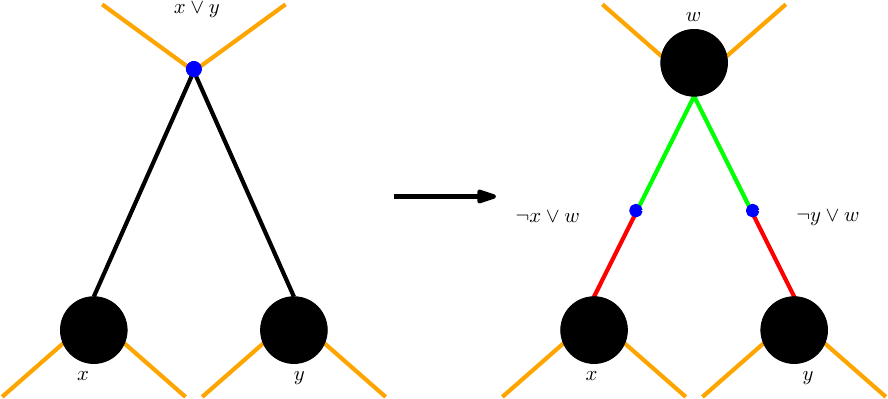}\caption{The planarity gadget. Left: orange is the VC-cycle in max2SAT. Right: Orange is the V-cycle in min2SAT.}\label{max2min_gadget}\end{figure}
\begin{figure}\centering\includegraphics[width=0.9\textwidth,page=1]{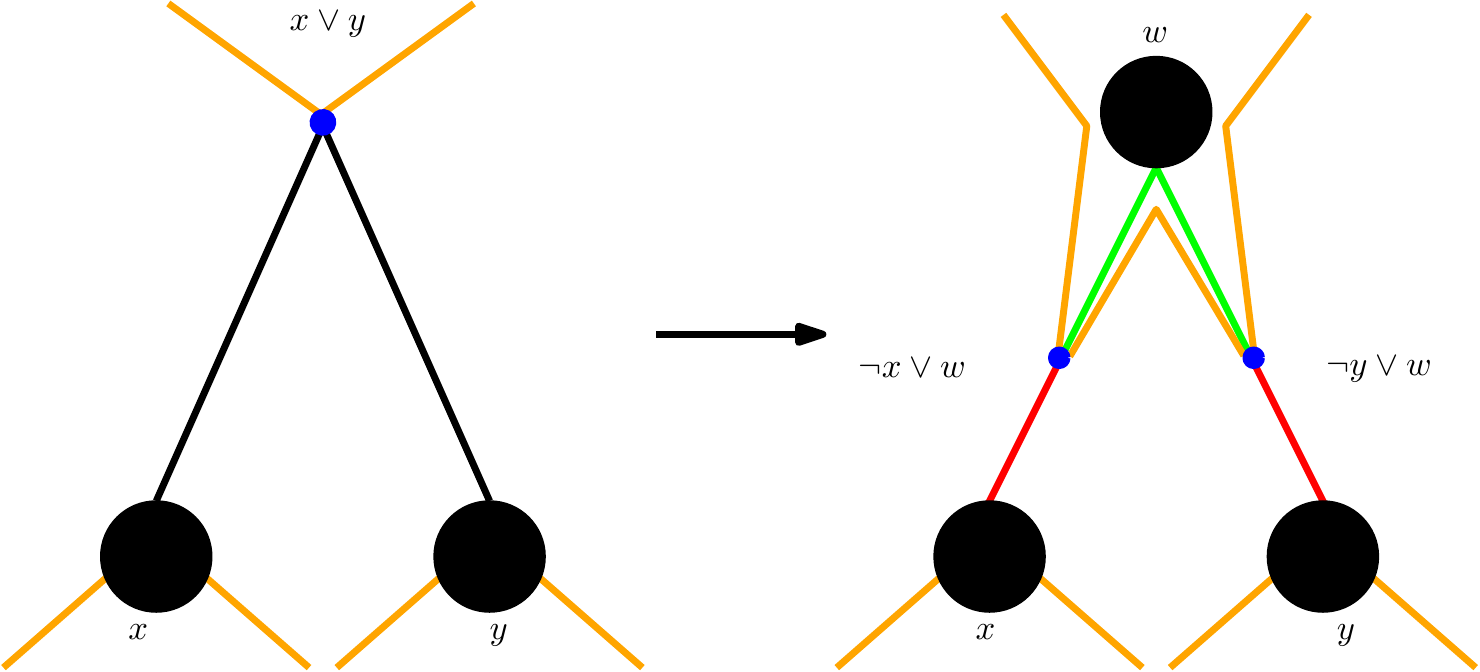}\\
\includegraphics[width=0.9\textwidth,page=2]{max2sat_to_min2satVC}\\
\vspace{20pt}
\includegraphics[width=0.9\textwidth,page=3]{max2sat_to_min2satVC}
\caption{The planarity gadget. Left: orange is the VC-cycle in max2SAT. Right: Orange is the VC-cycle in min2SAT. Top to bottom: the 3 ways of how the cycle goes through $x\vee y$ in VC-cycle max2SAT.}\label{max2minVC_gadget}\end{figure}
\end{proof}

A separable max2SAT with VC cycle which separates also the variables at the clauses can be solved in polynomial time:
\begin{theorem}\label{Kevin}Separable max2SAT with VC cycle, in which the VC cycle also separates the variables at the clauses (i.e., at each clause the connections from the two variables come from the different sides of the cycle), is polynomial-time solvable.\end{theorem}
\begin{proof}Let $I$ be an instance of max2SAT as in the statement of the theorem. The standard reduction from max2SAT to min2SAT \cite{minsatOrig} (see the proof of Theorem~\ref{t:opt2sat} above) preserves the planarity. Moreover, the obtained min2SAT instance $I$ is separable: the cycle separates the original literals from $I$ because $I$ is separable, and it separates the new literals in $I'$ because it separated the connections to clauses in $I$. By Theorem~\ref{t:separable:min} the minimum number of satisfied clauses in $I'$, and hence the maximum number of satisfied clauses in $I$ can be found in polynomial time.\end{proof}

\end{document}